%% file: main.tex
\documentclass[final]{l4dc2024}
\usepackage[utf8]{inputenc}
\usepackage{hyperref}
\usepackage{natbib}
\title[Multi-agent assignment via state augmented reinforcement learning]{Multi-agent assignment via state augmented reinforcement learning}
\usepackage{times}

\author{%
 \Name{Leopoldo Agorio} \Email{lca31@pitt.edu}\\
 \addr University of Pittsburgh, Pittsburgh, PA, USA
 \AND
 \Name{Sean Van Alen} \Email{sgv13@pitt.edu}\\
 \addr University of Pittsburgh, Pittsburgh, PA, USA
 \AND
 \Name{Miguel Calvo-Fullana} \Email{miguel.calvo@upf.edu}\\
 \addr Universitat Pompeu Fabra, Barcelona, Spain%
 \AND
 \Name{Santiago Paternain} \Email{paters@rpi.edu}\\
 \addr Rensselaer Polytechnic Institute, Troy, NY, USA%
 \AND
 \Name{Juan Andr\'es Bazerque} \Email{juanbazerque@pitt.edu}\\
 \addr University of Pittsburgh, Pittsburgh, PA, USA%
}
\usepackage{natbib}
\usepackage{bm}
\usepackage{graphicx}
\usepackage{amsmath,amssymb}
\usepackage{bbm}
\newtheorem{assumption}{Assumption}

\usepackage{dsfont}
\usepackage{algorithm}
\usepackage{tikz}
\usetikzlibrary{automata, positioning,arrows,arrows.meta,shapes.geometric}
\usepackage{pgfplots}

\usepackage{ifthen}
\newboolean{showcomments}
\setboolean{showcomments}{false}
\usepackage{todonotes}

\newcommand{\reviewer}[1]{  \ifthenelse{\boolean{showcomments}}
{\todo[inline,color=lime]{Reviewer: #1}}{}}
\newcommand{\response}[1]{  \ifthenelse{\boolean{showcomments}}
{\todo[inline,color=orange]{Response: #1}}{}}

\newcommand{\santiago}[1]{  \ifthenelse{\boolean{showcomments}}
{\todo[inline,color=cyan]{Santiago: #1}}{}}

\newcommand{\miguel}[1]{  \ifthenelse{\boolean{showcomments}}
{\todo[inline,color=green]{Miguel: #1}}{}}

\newcommand{\leopoldo}[1]{  \ifthenelse{\boolean{showcomments}}
{\todo[inline,color=yellow]{Leopoldo: #1}}{}}

\newcommand{\juan}[1]{  \ifthenelse{\boolean{showcomments}}
{\todo[inline,color=pink]{Juan: #1}}{}}

\DeclareMathOperator*{\argmax}{argmax}

\DeclareMathOperator{\subjectto}{subject\ to}

\begin{document}

\maketitle

\begin{abstract}%
  We address the conflicting requirements of a multi-agent assignment problem through constrained reinforcement learning, emphasizing the inadequacy of standard regularization techniques for this purpose. Instead, we recur to a state augmentation approach in which the oscillation of dual variables is exploited by agents to alternate between tasks. In addition, we coordinate the actions of the multiple agents acting on their local states through these multipliers, which are gossiped through a communication network, eliminating the need to access other agent states. By these means, we propose a distributed multi-agent assignment protocol with theoretical feasibility guarantees that we corroborate in a monitoring numerical experiment.
\end{abstract}

\begin{keywords}%
Constrained reinforcement learning, multi-agent assignment, monitoring problem.
\end{keywords}

\section{Introduction}
\input{01_intro}

\section{Problem formulation} \label{sec_probform}
\input{02_problem_formulation}
\section{State augmented Markov decision process}\label{sec_lagrangian}
\input{03_lagrangian.tex}
\section{Offline training}\label{sec_offline_training}
\input{04_offline_training.tex}
\section{Distributed online execution}\label{sec_online_execution}
\input{05_gossiping.tex}
\section{Convergence analysis}\label{sec_convergence}
\input{06_convergence_analysis}
\section{Numerical experiments}\label{sec_numerical}
\input{07_numerical}
%
\section{Conclusion}\label{sec_conclusion}
We showcased the efficacy of constrained RL in handling conflicting requirements via a state-augmented approach with oscillating dual variables.
By introducing a multi-agent system that coordinates through gossiped binary variables, we offer a solution for multi-agent assignment establishing almost sure feasibility guarantees in Theorem \ref{theorem:feasibility}, corroborated in a monitoring numerical experiment involving communication intermittence and realistic robot navigation dynamics.
%
%
\section*{Appendix}
\input{08_appendix}
\newpage
\acks{This work was supported in part by Spain's Agencia Estatal de Investigaci\'on under grant RYC2021-033549-I and the Mar\'ia de Maeztu Strategic Research Program under grant CEX2021-001195-M.}

\bibliography{references}

\end{document}

%% file: 01_intro.tex
Reinforcement Learning (RL) has experienced growing interest as a methodology for designing optimal controllers in cases where system dynamics are intractable or unknown \cite{kahn2017uncertainty}. It has attained remarkable achievements in recent years, with agents capable of playing Go \cite{silver2016mastering} and Poker \cite{brown2019superhuman}, applications in robotics \cite{kober2013reinforcement}, and in training recommendation systems and natural language models \cite{christiano2017deep}.

Incorporating constraints into the RL problem offers extra capabilities, such as ensuring safety in robot navigation \cite{miryoosefi2019reinforcement,paternain2022safe}, imposing physical conditions \cite{li2019constrained,gao2020batch,porteiro2024combined}, satisfy natural language rules \cite{ouyang2022training}, or in general, setting problem specifications. 
%
%
Constrained RL is commonly addressed through regularization to balance conflicting requirements,  weighting the rewards to create a multi-objective RL problem \cite{kahn2017uncertainty}. In this approach,  selecting the appropriate regularization weights via duality \cite{paternain2019constrained}, Bayesian rules \cite{hutter2002self}, or simply heuristics is crucial for achieving good performance \cite{tessler2018reward}. 

However popular, there are problems for which regularization approaches are insufficient to ensure feasibility \cite{calvo2023state}. In this paper, we focus on one such family of problems, assignment problems \cite{bertsekas1981new,mills2007dynamic,chopra2017distributed}, in a dynamic multi-agent setup in which several tasks are assigned to a team of agents.
We will argue that regularized formulations of these multi-agent assignment RL problems fail to produce policies satisfying the conflicting constraints, lacking convergence.
To overcome this challenge, we leverage a recent state augmentation approach, where the Lagrangian weights are reinterpreted as states of an augmented Markov Decision Process (MDP) \cite{calvo2023state}. Under this model, the multipliers are not driven to convergence to a dual optimum. Instead, they are allowed to oscillate to induce alternating policies that satisfy the conflicting tasks sequentially.

We further demonstrate that these multipliers can be used for coordinating multiple agents, especially when many conflicting tasks cannot be achieved by a single agent or several ones acting independently, a state-of-the-art challenge in multi-agent reinforcement learning \cite{zhang2021multi}.    
In this context, we propose an assignment protocol to coordinate a team of agents running pre-trained distributed policies, where each agent might have only local information about the state of the system. For a fully distributed algorithm, we use the gossiping framework \cite{aysal2009broadcast} to share the multipliers across a communication network established by the agents, and we prove consensus in finite time, taking advantage of the binary nature of the assignment variables.

 Our state augmentation and multi-agent assignment protocol come with theoretical guarantees of almost sure feasibility, which we corroborate with numerical experiments involving a team of robots that must patrol an area of interest alternating between different zones.

%% file: 02_problem_formulation.tex
Consider $N$ agents interacting with an environment that is modeled as an MDP with transition probabilities $\mathbb{P}\left(S_{t+1}\mid S_{t},A_{t}\right)$. For each time index $t=0,1,2,\ldots$, variables $S_{t}=(S_{t1},\ldots,S_{tN})$ and $A_{t}=(A_{t1},\ldots,A_{tN})$ collect the states $S_{tn}\in \mathcal S$ and actions $A_{tn}\in \mathcal A$ for all agents. The team receives $M$ rewards $r_m(A_t,S_t)$ associated with $M$ different tasks that need to be satisfied collectively. We consider a family of assignment problems in which rewards are given by
\begin{equation}\label{eqn_patrolling_reward} 
    r_m(S_t,A_t) = \max_{n=1,\ldots,N} \mathds{1}[S_{tn}\in \mathcal{S}_m]=\left\{\begin{array}{ll}1 & \text{if } \exists n\in \{1,\ldots,N\}: S_{tn}\in \mathcal{S}_m\\0 & \text{otherwise}\end{array}\right.,
\end{equation}
defined in terms of $M$ regions $\mathcal S_m\subset \mathcal S$ of the state space. Then, we set the requirements as $M$ thresholds $(c_1,\ldots,c_M)\in[0,1]^M$ on the average rewards and formulate a constrained Markov decision problem with  
constraints  $V_m(\pi)= \lim_{T\to \infty}\frac{1}{T} \mathbb E_{S_t, A_t\sim \pi} \left [ \sum_{t=0}^{T-1} r_m(S_t,A_t) \right ] \geq c_m, \ m= 1, \hdots ,M$.
%
The average is taken over the transition probabilities and the probability distribution defining the control policy $\pi(A_t\mid S_t)$  of the actions of all agents given their states. To avoid the exponential growth of the policy with the number of agents, we impose a separate design $\pi(A_t\mid S_t) = \prod_{n=1}^N \pi_n(A_{tn}\mid S_{tn})$, in which each agent decides its action taking only into account its local state and not the state of other agents.  
%
 %
It is important to remark that even if the policies are local, the agents are still coupled by their dynamics and joint rewards. 
\begin{example}\label{example}
As a practical example of the assignment problem described by \eqref{eqn_patrolling_reward}, consider a team of $N$ robots that are required to monitor $M>N$ areas,  visiting each area $\mathcal S_m$ a fraction of time $c_m$. These conflicting requirements may not be achievable by a single agent, particularly when $\sum_{m=1}^M c_m>1$, and not even by multiple agents acting independently following the same optimal single-agent policy. Therefore, coordination is essential to solve this problem successfully.
\end{example}

A task reward $r_0(S_t,A_t)$ can be added, together with its corresponding optimization objective  $V_0(\pi)=\lim_{T\to \infty}\frac{1}{T} \mathbb E_{S_t, A_t\sim \pi} \left [ \sum_{t=0}^{T-1} r_0(S_t,A_t) \right ]$. This would be useful, for instance, to set a secondary goal, such as spending the least amount of energy on the multi-robot surveillance task in Example \ref{example}. However, in the following, we will focus on the primary goal of satisfying the specifications, setting  $V_0(\pi)=r_0=0$. 
Specifically, the feasibility problem we aim to solve is
\begin{align}
P^\star =&\max_{\pi_1,\ldots,\pi_N} 0,\quad 
\subjectto \lim_{T\to \infty}\frac{1}{T} \mathbb E_{S_t, A_t\sim \pi} \left [ \sum_{t=0}^{T-1} r_m(S_t,A_t) \right ] \geq c_m, \ m= 1, \hdots ,M,\label{eqn_crl}
\end{align}
with  $r_m(S_t,A_t)$ as in \eqref{eqn_patrolling_reward} and separate  policy structure $\pi(A_t\mid S_t) = \prod_{n=1}^N \pi_n(A_{tn}\mid S_{tn})$.

%% file: 03_lagrangian.tex
We start this section by defining the Lagrangian associated with problem \eqref{eqn_crl}.
\begin{equation}\label{eqn_lagrangian}
\mathcal{L} (\pi,\lambda) = \lim_{T\to \infty}\frac{1}{T} \mathbb E_{S_t, A_t\sim \pi} \left [ \sum_{t=0}^{T-1} \sum_{m=1}^M \lambda_m\left(r_m(S_t,A_t)-c_m\right) \right ], 
\end{equation}
where in the previous expression $\lambda \in\mathbb{R}_+^M$ are the dual variables. 
The previous expression can be simplified by defining the following reward parameterized by $\lambda$
\begin{equation}\label{eqn_reward_lambda}
    r_\lambda(S_t,A_t) :=  \sum_{m=1}^M \lambda_m (r_m(S_t,A_t)-c_m).\\ 
\end{equation}
%
Under the reward structure of \eqref{eqn_patrolling_reward}, standard regularized methodologies are unable to produce feasible policies \cite[Proposition 1]{calvo2023state}. Hence, we aim to solve problem \eqref{eqn_crl} by a two-step process. First, in an offline training stage, we will learn the family of policies parameterized by $\lambda$ that solve the regularized version of \eqref{eqn_crl}, i.e.,
\begin{align}\label{eqn_optimal_set}
&\Pi(\lambda) = \argmax_{\pi_1,\ldots,\pi_N} \lim_{T\to \infty}\frac{1}{T} \mathbb E_{S_t, A_t\sim \pi} \left [ \sum_{t=0}^{T-1} r_{\lambda}(S_t,A_t) \right ].
\end{align}
Then, in the deployment stage, each agent will sample from policies $\Pi(\lambda)$, with $\lambda$ following dual gradient descent. To be formal, let us define the dual function as
%
 $  d(\lambda) := \max_{\pi} \mathcal{L}(\pi(\lambda),\lambda),$ 
%
where $\pi(\lambda) \in \Pi(\lambda)$. 
{For the optimal multipliers, we minimize $d(\lambda)$ over $ \lambda_m\geq 0$. 
By Danskin's Theorem \cite{danskin1966theory}, we substitute the constraints evaluated at the optimal policy for the gradient of $d(\lambda)$ to obtain the dual update
%
    $\lambda^{k+1} = \left[\lambda^{k}-\eta\left(V(\pi(\lambda^{k})-c\right)\right]_+,$
}
where $[\cdot]_+$ denotes the projection on the non-negative orthant \textcolor{black}{and $c=(c_1,\ldots,c_M)^T$.} Section \ref{sec_online_execution} will incorporate a truncated stochastic version of these dual                  dynamics into an augmented MDP with state $(S_t,\lambda^{k})$ in two timescales. 

%% file: 04_offline_training.tex
In this section, we parameterize the policy of each agent by a vector $\theta_n \in \mathbb{R}^{p_n}$ to deal with the continuous augmented state-space $\mathcal{S}\times {\mathbb{R}^M_+}$. Hence, each agent's action is drawn from a distribution~$\pi_{\theta_n}(A_{tn} \mid S_{tn}, \lambda)$. As in Section \ref{sec_probform},  the joint probability distribution is $\prod_{n=1}^N\pi_{\theta_n}(A_{tn}\mid S_{tn},\lambda)$.  Thus, the policy gradient gradient Theorem \cite{sutton2000policy} takes a form that allows for distributed estimation (see Proposition \ref{prop_gradients}). Before stating this result, we define some key concepts.

For fixed multipliers and parameters $\theta \in \mathbb{R}^p$, where $p = \sum_{n=1}^N p_n$, we define the following occupancy measure for the joint state and actions of all agents,
%
 $   \rho_{\theta,\lambda}(s,a) = \lim_{t\to \infty} \mathbb{P}(S_t=s,A_t=a).$
%
Likewise,  $   Q_{\lambda}^{\pi_\theta}(s,a) = \sum_{t=0}^{\infty} \mathbb{E}_{S_t,A_t\sim \pi_\theta}\left[r_\lambda(S_t,A_t)-\mathcal L(\pi_\theta,\lambda)\big| S_0=s,A_0 =a\right],$  defines the state-action value function,
%
where $r_\lambda(s,a)$ and $\mathcal{L}(\theta,\lambda)$ are the functions defined in \eqref{eqn_reward_lambda} and \eqref{eqn_lagrangian} respectively. Notice that we omitted the dependency of  $\pi_\theta$  on the multipliers $\lambda$  to simplify the notation. We further define the state-action value function for agent $n$ as
\begin{equation}\label{eqn_q_others}
    Q_{n,\lambda}^{\pi_\theta}(s,a) = \mathbb{E}_{(S_t,A_t)~\sim\rho_{\theta,\lambda}}\left[Q_{\lambda}^{\pi_\theta}(S_t,A_t)\mid S_{tn}=s, A_{tn}=a \right].
\end{equation}
From the perspective of agent $n$, this quantity is the average of the expected return, assuming that all other agents are following policy $\pi_\theta$. \textcolor{black}{From this perspective,  $Q_{n,\lambda}^{\pi_\theta}(s,a)$ is estimated per agent by averaging the global rewards that result from their local states and actions, or modeled as a parametric function of these local entries.} We are now in conditions of providing the specific form that the policy gradient Theorem takes under the conditions described in this section.
\begin{proposition}\label{prop_gradients}
Assume that the policy of each agent is parameterized by a vector $\theta_n\in\mathbb{R}^n$ and that $A_n\sim \pi_{\theta_n}(\cdot\mid S_n,\lambda)$. Let $\mathcal{L}(\pi_\theta,\lambda)$ and $Q_{n,\lambda}^{\pi_\theta}(S_n,A_n)$ be the functions defined in \eqref{eqn_lagrangian} and \eqref{eqn_q_others} respectively. Then, it follows that
%
  $  \nabla_{\theta_n} \mathcal{L(\pi_\theta,\lambda)} = \mathbb{E}_{S_n,A_n \sim \rho_{\theta,\lambda} }\left[\nabla_{\theta_n }\log \pi_{\theta_n}(A_n\mid S_n,\lambda)Q_{n,\lambda}^{\pi_\theta}(S_n,A_n)\right].$
\end{proposition}
Leveraging on the local gradients in Proposition \ref{prop_gradients}, we  optimize a  parametric version of \eqref{eqn_optimal_set}
%
\begin{align}\label{eqn_optimal_trained_policy}
&\left(\pi_{\theta_1},\ldots,\pi_{\theta_N}\right) = \argmax_{\pi_{\theta}}  \mathbb E_{S_t, A_t\sim \pi_\theta} \left [ \frac{1}{T_0}\sum_{t=0}^{T_0-1} r_{\lambda}(S_t,A_t) \right ].
\end{align}
were the optimization variable was written as $\pi_\theta=\left( \pi_{\theta_1},\ldots,\pi_{\theta_N}\right)$ to reduce notation. Notice that, in addition to the parametric expansion of the policies, we introduced a second theoretical variation by removing the limit and adhering to a finite time horizon. This modification is key for the practical implementation of the algorithm. For convergence analysis, we still link it to our original formulation in \eqref{eqn_crl} using the following assumption.
\begin{assumption}\label{assumption_consistency}\textbf{Estimation consistency.}
The truncating error in \eqref{eqn_consistency} satisfies  $\lim_{T_0\to\infty} \left \|e_{T_0} \right\| =0.$ %
\begin{equation}\label{eqn_consistency}
    \mathbb{E}\left[\frac{1}{T_0}   \sum_{t=k T_0}^{(k+1)T_0}  \left(r(S_t,A_t)-c\right) \right]=\mathbb{E}\left[\lim_{T\to\infty}\frac{1}{T}  \textcolor{black}{\sum_{t=0}^{T-1} } \left(r(S_t,A_t)-c\right) \right] +\textcolor{black}{ e_{T_0}}.
\end{equation}
\end{assumption}

\noindent{\textbf{Remark}} Training for \eqref{eqn_optimal_trained_policy} requires that all agents interact  to learn from  rewards \eqref{eqn_reward_lambda} that are jointly activated  via \eqref{eqn_patrolling_reward}. This training phase is designed to run offline, with multiple agents and a common $\lambda$ drawn randomly at the start of an episode and fixed until the end of it.  Once each policy $\pi_{\theta_n}(A_{tn},\mid S_{tn},\lambda)$ in \eqref{eqn_optimal_trained_policy} has been trained, its online execution depends only on the local state $S_{tn}$ of the agent, and a common vector $\lambda$ which varies as agents enter and exit the zones. Thus, the online coordination between agents depends on sharing the multipliers $\lambda$ only, as described next.

%% file: 05_gossiping.tex
Aiming for a practical online implementation of a feasible policy for \eqref{eqn_crl}, we introduced a parametric training strategy in \eqref{eqn_optimal_trained_policy} and truncated the episode.  There are still two practical issues to be addressed before running the policies \eqref{eqn_optimal_trained_policy} in the execution phase. These are computing the expectation over the transition dynamics in the dual update, and assuming that all agents can see the same  $\lambda^{k}$.

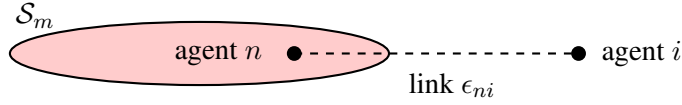
\begin{figure}
    \centering
    \input{gossip}
    \vspace{-0.5cm}
    \caption{ \small Gossip protocol between two agents $n$ and $i$ communicating through the link $\epsilon_{ni}$. Both agents aim to know if the shaded zone $\mathcal S_m$  is occupied at time $\tau=0$, as defined by   $\max\{\mathds 1[S_{\tau i}\in \mathcal S_m],\mathds 1[S_{\tau n}\in \mathcal S_m] \}$. Since agent $n$ is in $\mathcal S_m$ at time $t=0$, it knows the actual reward. Hence $\hat R_{1,\tau,m,t}=1$ for all $t\geq 0$. Instead, agent $i$ is outside $\mathcal S_m$ at $t=0$ so that $R_{i,\tau,m,0}=0$ is incorrect, but is corrected at time $t=1$ by gossiping from agent $n$, so that $R_{i,\tau,m,t}=1$  for all $t\geq 1$. \vspace{-0.8cm}}
    \label{fig:gossip}
\end{figure}

We resort to a stochastic approximation of the gradient to address the first issue. As introduced in Section \ref{sec_lagrangian}, a dual gradient descent update over the  Lagrange multipliers for problem \eqref{eqn_crl} amounts to
%
$\lambda_m^{k+1} =\left[\lambda_m^{k} - \eta \lim_{T\to \infty}\frac{1}{T} \mathbb E_{S_t, A_t\sim \pi} \left [ \sum_{t=0}^{T-1}  \left(r_m(S_t,A_t)-c_m\right) \right ]\right]_+. $
%
For a data-driven implementation, we run stochastic gradient descent   \cite{dvoretsky1955stochastic}, dropping the expectation and replacing it with reward samples. As we did for training, we also set a finite horizon $T_0$,
\begin{align}\label{eqn_stochastic_dual}
\hat \lambda_m^{k+1} &=\left[\hat \lambda_m^{k} - \frac{\eta}{T_0}   \sum_{t=k T_0}^{(k+1)T_0-1}  \left(r_m(S_t,A_t)-c_m\right) \right]_+.
\end{align}
The trade-off for choosing the roll-out horizon  $T_0$ entails keeping the induced error in \eqref{eqn_consistency} sufficiently small while ensuring that the multipliers are updated sufficiently often.

We propose a gossiping protocol to solve the second issue.
Notice that agents must know the global rewards $r_m(S_t,A_t)$ to compute the update in \eqref{eqn_stochastic_dual}.
The following analysis aims for a distributed implementation of \eqref{eqn_stochastic_dual} in which each agent updates a local copy of $\lambda^{(k)}$ using their own views of the global rewards, which are agreed with the team through local exchanges with neighboring agents over a communication network.
To formalize these ideas, consider agents $n\in \mathcal V=\{1,\ldots,N\}$ as nodes of a \emph{connected} graph $\mathcal G(\mathcal V,\mathcal E)$ with communication links $\mathcal E$ such that $e_{n,i}\in \mathcal E $ if agents $n$ and $i$ are within direct communication range. Define the neighborhood $\mathcal N_n=\{i\in \mathcal V: e_{n,i}\in E\}$  of agents communicating with  $n$,
%
 the distance $d^\star_{n,i}
$ as the minimum number of hops required to communicate two remote nodes $n$ and $i$, and
 the graph diameter $d(\mathcal G):=\max_{n,m\in V}d^\star_{n,m}$  as the maximum distance between nodes across the graph.

The global occupation  of zone  $\mathcal S_m$ at time $\tau$ takes a binary value $r_m({S_\tau,A_\tau})\in\{0,1\}$.   Agent $n$ computes a sequence of local estimates  of $r_m({S_\tau,A_\tau})$ for times $t\geq \tau$.
As shown in  Fig. \ref{fig:gossip}, these estimates $\hat R_{m,\tau,n,t}$ are initialized by agent $n$ at time $t=\tau$  using its local occupancy index; i.e.,
   \textcolor{black}{ $\hat R_{m,\tau,n,\tau}=\mathds 1[S_{\tau n}\in \mathcal S_m]$,}
 and then agreed with the neighbor agents $i\in \mathcal N_n$ including $n$ itself via $\hat R_{m,\tau,n,t}=\max_{i\in\mathcal N_n\cup \{n\}} \textcolor{black}{\hat R_{m,\tau,n,t-1}}$ at each subsequent time $t>\tau$.
Through this gossip algorithm, binary data $\mathds 1[S_{tn}\in \mathcal S_m]$ percolates across the network via local exchanges, reaching consensus within a time interval no longer the network diameter, as shown next

\begin{lemma}\label{lemma_gossip}
The gossip estimate $\textcolor{black}{\hat R_{m,\tau,n,t}}$ underestimates $r_m(S_\tau,A_\tau)$; i.e., $\textcolor{black}{\hat R_{m,\tau,n,t}}\leq r_m(S_\tau,A_\tau)$ with equality if the gossiping time $t$ is large enough so that remote data arrive through local exchanges across the network; i.e.,
$\textcolor{black}{\hat R_{m,\tau,n,t}}= r_m(S_\tau,A_\tau),\ \text{ for all }n\in V \text{ and }t\geq \tau+ d(\mathcal G).$
\end{lemma}

This lemma is to be used for convergence analysis, but also for algorithm construction, since it confirms that the estimates $\hat R_{m,\tau,n,t}$ do not need to be updated beyond $t\geq \tau+d$, where $d\geq d(\mathcal{G})$ is any overestimate of the network diameter. Thus, putting   $t=\tau$, $t\in\{\tau+1,\tau+d-1\}$ and $t>\tau+d$ together,  we update the distributed reward estimates as
\begin{align}\label{eqn_gossip_r}
    	\hat R_{m,\tau,n,t}&=\begin{cases}
    	\hat R_{m,\tau,n,t-1},&\quad \tau=(k-1)T_0,\ldots,t-d-1 \\
    	\textcolor{black}{\max_{p\in\mathcal N_n\cup \{n\}} \hat R_{m,\tau,p,t-1}},&\quad \tau=t-d,\ldots,t-1\\
    	\mathds 1\{S_{nt}\in\mathcal S_m\},&\quad \tau=t
    	\end{cases}
	\end{align}
We can substitute them into the dual update \eqref{eqn_stochastic_dual} to obtain a distributed algorithm for the multipliers.
\begin{algorithm}[t]
\For{$k=0,1,\ldots$}{
  Set local policy $\pi_n=\pi_\theta(\lambda^{k}_n)$ as in \eqref{eqn_optimal_trained_policy} pre-trained offline using local gradients.\;\\
  \For{$t=k T_0:(k+1)T_0-1$}{
  	Take action, transition, collect local rewards $\mathds{1}\{S_{tn}\in\mathcal S_m\}$ \\\;
  	\emph{Network gossiping:} Update $\hat R_{m,\tau,n,t}$ according to \eqref{eqn_gossip_r}
	  }
   Compute the stochastic dual gradients
   $$\textcolor{black}{g_{nmk}=\frac{1}{T_0}   \sum_{\tau=(k-1) T_0}^{ kT_0-1}  \left(\hat R_{m,\tau,n,(k+1) T_0} -c_m\right)},\quad \textcolor{black}{g_{nm(k+1)}=\frac{\eta}{T_0}   \sum_{\tau=kT_0}^{(k+1)T_0-1}  \left(\hat R_{m,\tau,n,(k+1) T_0} -c_m\right)}$$
  Update the multipliers $\lambda^{k}_{m,n}=\left[\lambda^{k-1}_{m,n} -\eta \textcolor{black}{g_{nmk}}
   \right]_+,\quad \lambda^{k+1}_{m,n}=\left[\lambda^{k}_{m,n}
  -\eta \textcolor{black}{g_{nm(k+1)}} \right]_+$
}
\caption{Distributed multi-agent dual update and policy execution}
\label{algo:alg_main_algorithm}
\end{algorithm}
Algorithm \ref{algo:alg_main_algorithm} not only indicates how to update the local estimates $\lambda_{m,n}$ via network gossiping of the rewards but also gives each agent the procedure to run its fully distributed policy  $\pi_\theta(\lambda^{k}_n)$ in the execution phase.
The reason to keep two copies of $\lambda_{m,n}$ by Algorithm \ref{algo:alg_main_algorithm} is to wait until agents reach consensus on the rewards via Lemma \ref{lemma_gossip}; see the proof of Proposition \ref{prop_gossip} for more details.  In the proof of Lemma \ref{lemma_gossip} is shown that $\hat \lambda_{m,n,t}^{k}$ reaches steady state  and can be discarded once $t\geq t_{k+1}$ since, according to \eqref{eqn_distributed_dual_simple}, it is not needed  for recalculating  $\hat \lambda_{m,n}^{k+1}$ any more.

Next, we show that running the gossiping and distributed policies ensures that the constraints in \eqref{eqn_crl} are satisfied almost surely.  Building up to this main result, we first characterize the multipliers mismatch, and then add two working assumptions.

\begin{proposition}\label{prop_gossip}
Let $T_{k+1}:=T_0 (k+1)$ the time at the end of rollout $k+1$. Assume that $T_0\geq d(\mathcal G)$. 
Then, the error on the multipliers is bounded by
 $   \lambda_{m}^{k+1}-\lambda_{m,n}^{k+1}\leq \frac{\eta}{T_0}d(\mathcal G)$
\end{proposition}

\begin{assumption}\label{assumption_noforces}\textbf{No repelling forces}
\textcolor{black}{The underlying Markov decision process is such that, given $m$, there exists a feasible policy  $\pi_n$ satisfying $r_{m}(S_t,A_t)=1$ for all $t$.}
\end{assumption}

\textcolor{black}{The previous assumption implies that an agent
 can be stationed at a particular zone $S_{nt}\in \mathcal S_m$. This will be used in the proof of Lemma \ref{lemma:bounded_multipliers}.}

\begin{assumption}\label{assumption_representation}\textbf{Representation capacity}
The policy parameterization is dense enough; i.e.,
 $\exists\beta>0$, such that for each  $\pi(\lambda)\in \Pi(\lambda)$  there is a $ \pi_\theta(\lambda)$ satisfying
%
\textcolor{black}{$\left|\sum_{m=1}^M \lambda^T \left(V_m( \pi(\lambda))-V_m(\pi_\theta(\lambda))\right)\right|\leq \beta \|\lambda\|,$}
%
where $V_m(\pi(\lambda))$ and $V_m(\pi_\theta(\lambda))$ are the values resulting from $\pi(\lambda)$ and $\pi_\theta(\lambda)$, respectively.
\end{assumption}

\begin{theorem}\label{theorem:feasibility}
Let Assumptions \ref{assumption_consistency}--\ref{assumption_representation} hold. Consider specifications $c_{\max}=\max_{m=1,\ldots,M} c_m <1$,  and $\sum_{m=1}^Mc_m \leq N-1$ so that  $\delta_c=(1-c_{\max})/\sqrt{M}>0$. If  $\beta+ \frac{M}{T_0}d(\mathcal G)\eta +\epsilon_{T_0}+\eta/2 < \delta_C$,  the trajectories generated by Algorithm \ref{algo:alg_main_algorithm} are feasible with probability one, i.e., for all $m=1,\ldots,M$
\begin{equation}\label{eqn_as_feasibility}
    \liminf_{T\to \infty}\frac{1}{T} \sum_{t=0}^{T-1}  r_m(S_t,A_t)= c_m, \mbox{ a.s.}
    \end{equation}
\end{theorem}

The condition in  Theorem \ref{theorem:feasibility} can be satisfied by designing a sufficiently powerful parameterization or neural network, a sufficiently long time horizon, and then choosing a small step-size.

%% file: gossip.tex
\begin{tikzpicture}[scale=0.42]

\filldraw[color=black, fill=red!20, thick](0,0) ellipse (6cm and 1cm);
\filldraw[color=black, fill=black, thick](3,0) circle (0.2);
\filldraw[color=black, fill=black, thick](12,0) circle (0.2);
\draw[color=black, thick, dashed] (3,0) -- (12,0);

\draw node  at (-5.2,1.2) {$\mathcal S_m$};
\draw node  at (0.6,0) {agent $n$};
\draw node  at (14,0) {agent $i$};
\draw node  at (8,-1) {link $\epsilon_{ni}$};

\end{tikzpicture}

%% file: 06_convergence_analysis.tex
Define
$\hat g^{k}=\frac{1}{T_0}\sum_{\tau=(k-1)T_0}^{kT_0-1} \left(r(S_{\tau},A_{\tau})-c\right).$
%
 This vector contains the average constraint violation during execution between two updates of the Lagrange multipliers $\lambda^k$. In the following lemma, we establish that $\hat g^{k}$ is in expectation similar to the gradient of the dual function.

\begin{lemma}\label{lemma:error_gradient}
With $\mathcal F_k$ being the filtration for the process $\hat g^k,$ assumptions \ref{assumption_consistency}--\ref{assumption_representation} ensure that for every $\epsilon_{T_0}>0$, there exists a $T_0$ such that
 $  \textcolor{black}{(\lambda^k)^T \left(\nabla_\lambda d(\lambda^k)-\mathbb{E}_{\pi_\theta}\left[\hat g^k|\mathcal F_k\right]\right)\leq \|\lambda^k\|\left(\beta +  \frac{M}{T_0}d(\mathcal G)\eta+\epsilon_{T_0}\right).}$
\end{lemma}

\begin{proof}
\textcolor{black}{
Write $(\lambda^k)^T \left(\nabla_\lambda d(\lambda^k)-\mathbb{E}_{\pi_\theta}\left[\hat g^k|\mathcal F_k\right]\right) =
(\lambda^k)^T\left(V-V_N + V_N-V_{T_0}+V_{T_0}-V_\theta\right)$  starting from  $V=V(\pi_1(\lambda^k),\ldots,\pi_N(\lambda^k)))=\nabla_\lambda d(\lambda^k)$ with policies in \eqref{eqn_optimal_set}, then  adding and subtracting $V_N=V(\pi_1(\lambda_1^k),\ldots,\pi_N(\lambda_N^k)))$ such that the agents optimize using distributed local copies of $\lambda^k$,   and $V_{T_0}$ where additionally the objective is finite horizon.  The value   $V_{\theta}=\mathbb{E}_{\pi_\theta}\left[\hat g^k|\mathcal F_k\right] =V(\pi_{\theta_1}(\lambda_1^k),\ldots,\pi_{\theta_{N}}(\lambda_N^k))$ on the right involves the algorithmic policies such that the agents optimize parametric policies over a finite horizon with distributed copies of $\lambda^k$.  The proof follows from the bounds in Assumption \ref{assumption_representation}, Proposition
\ref{prop_gossip}, and Assumption \ref{assumption_consistency}. For the difference $(\lambda^k)^T(V-V_N)$ we need to accumulate the bounds in Proposition \ref{prop_gossip} per zone, and then use  Lipschitz continuity with respect to the multipliers, with unit constant. The continuity can be argued as in the proof of Lemma \ref{lemma:bounded_multipliers} using Assumption \ref{assumption_noforces} and noticing that when the multipliers change their order  the assignment is altered, and the difference in the weighted value is equal to the difference of the multipliers.   }
\end{proof}

\begin{lemma}\label{lemma:bounded_multipliers}
Assumption \ref{assumption_noforces} with $\delta_C$  as in Theorem \ref{theorem:feasibility} results in
 $ \sum_{i=1}^M\lambda_m (V_m-c_m)\geq \delta_c \|\lambda\|. $
 \end{lemma}
\begin{proof}
Without loss of generality, let us sort the multipliers $\lambda_1\geq \lambda_2\geq\ldots\geq \lambda_M$. Then, one policy that optimizes the Lagrangian would sit one agent at each zone $m=1\ldots,N$. Therefore $V_m=1$ for  $m\leq N$ and $V_m=0$ for  $m>N$. Hence
%
 $   \sum_{m=1}^M \lambda_m  (V_m-c_m)= \lambda_1 (1-c_1)+\sum_{m=2}^N \lambda_m (1-c_m)+ \sum_{m=N+1}^M \lambda_m  (0-c_m)$ $\geq (1-c_{\max}) \lambda_1+ \sum_{m=2}^M \lambda_n (1-c_m)+ \sum_{m=N+1}^M \lambda_n  (0-c_m)$
    $= (1-c_{\max})\|\lambda\|_{\infty}+ \left(N-1-\sum_{m=2}^M c_m\right) \lambda_n$ $\geq (1-c_{\max})\|\lambda\|_{\infty}$ $\geq (1-c_{\max})\frac{\|\lambda\|}{\sqrt{M}}$.
\end{proof}

With this result, we establish conditions in the following lemma that the norm of the multipliers is a supermartingale. This will be instrumental in proving that the multipliers are bounded, and from there, that the constraints are satisfied.
\begin{lemma}\label{lemma:martingale}
 If   $\beta+ \frac{M}{T_0}d(\mathcal G)\eta +\epsilon_{T_0}+\eta/2 < \delta_C$ and $\left\|\lambda^{k} \right\|\geq 1$, then $\left\|\lambda^{k} \right\|^2$ is a supermartingale.
\end{lemma}
\begin{proof}
Using the update of the multipliers given in \eqref{eqn_stochastic_dual} and using the non-expansiveness of the projection, it follows that
$\Lambda^{k+1}\leq\|\lambda^k\|^2-2\eta (\lambda^{k})^T\hat g^{k}+\eta^2\|\hat g^{k}\|^2\leq \Lambda^k-2\eta(\lambda^k)^T\hat g^{k}+\eta^2.$
%
where we defined $\Lambda^{k}=\|\lambda^k\|^2$ and  used $\left\|\hat g^k \right\|^2 \leq 1$. 
    Therefore, to ensure that the process $\Lambda^k$ is a super-martingale, we need to prove   that
%
  $      \mathbb{E}\left[\left(\lambda^k\right)^T\hat g^k|\mathcal F_k\right]=\left(\lambda^k\right)^T \mathbb{E}\left[\hat g^k|\mathcal F_k\right]>\eta/2,$
where the expectation is taken with respect to the distributed policy \eqref{eqn_optimal_trained_policy} being executed.
%
%
To this end,  consider
    $\left(\lambda^k\right)^T\mathbb{E}\left[\hat g^k|\mathcal F_k\right]=\left(\lambda^k\right)^T \left(\nabla d (\lambda_k)+ \mathbb{E}\left[\hat g^k|\mathcal F_k\right]-\nabla d (\lambda_k)  \right)$
        $\geq \left(\lambda^k\right)^T\nabla d(\lambda_k)
    - \left(\beta+\frac{M}{T_0}d(\mathcal{G})\eta+\epsilon_{T_0}\right)\left\|\lambda^k\right\|,$
where the inequality follows from the Cauchy-Shwartz inequality and the bound in Lemma \ref{lemma:error_gradient}.
Combining the previous inequality along with the result of Lemma \ref{lemma:error_gradient} , and the hypothesis of this lemma, i.e.,  $\beta+\frac{M}{T_0}d(\mathcal G)\eta +\epsilon_{T_0}+\eta/2 < \delta_c$,  it follows for all $\lambda^k$
$    \left(\lambda^k\right)^T \mathbb E_{\pi^\star}\left[\hat g^k|\mathcal F_k\right]
\geq  \eta/2 \left\|\lambda^k\right\|\geq \eta/2$, where we use the hypothesis  $\left\|\lambda^k\right\|\geq 1$.
\end{proof}
\begin{proof}[\textbf{of Theorem \ref{theorem:feasibility}}] Due to space limitations we provide a sketch of the proof. Since $\Lambda^k= \left\|\lambda^k\right\|^2$ is a supermartingale, which is a generalization of a non-increased sequence, the multipliers must be bounded. More formally, the sequence $p(\lambda^k\mid \lambda^0)$ is \emph{tight}. That is, for any $\delta>0$, it exists a compact set $\mathcal{K}_\delta$ such that for every $k\geq 0$, $\mathbb{P}(\lambda^k\in \mathcal{K}_\delta)>1-\delta$.
For simplicity in the sketch of this proof, we will consider bounded multipliers instead of tight ones and show how this implies feasibility. For a similar argument on how $\Lambda^k$ being a supermartingale implies tightness and feasibility see Lemma 4 and Proposition 3 in \cite{calvo2023state}. Following the case of bounded multipliers, recall the dual descent update in \eqref{eqn_stochastic_dual}
%
$\hat \lambda_m^{k+1} =\left[\hat \lambda_m^{k} - \frac{\eta}{T_0}   \sum_{t=k T_0}^{(k+1)T_0}  \left(r_m(S_t,A_t)-c_m\right) \right]_+ \geq \hat \lambda_m^{k} - \frac{\eta}{T_0}   \sum_{t=k T_0}^{(k+1)T_0}  \left(r_m(S_t,A_t)-c_m\right). $
%
Applying this inequality recursively, it follows that
%
$\hat \lambda_m^{k+1} \geq \hat \lambda_m^{0} - \frac{\eta}{T_0}   \sum_{t=0}^{(k+1)T_0}  \left(r_m(S_t,A_t)-c_m\right). $
%
Taking the limsup in both sides of the previous inequality, and using that the multipliers are bounded it implies that  
%
%
 $   \limsup_{k\to\infty} -\sum_{t=0}^{(k+1)T_0}  \left(r_m(S_t,A_t)-c_m\right)$ $= L< \infty.$
%
Defining $T=({k+1})T_0$ the previous limit is equivalent to the desired \eqref{eqn_as_feasibility}.
%
%
    \end{proof}

%% file: 07_numerical.tex
\begin{figure}[t]
\floatconts
{fig:numericals}
{
\caption{\small
Simulation results obtained after executing for $200{,}000$ iterations, a two-agent policy trained to monitor the regions shown in (e), with requirements $c=[0.3,0.3,0.3,0.3]$. Colors in (a)--(c) correspond to the matching colored region in (e). A subset policy for $\lambda=[5,2.5,0,5]$ and $100$ steps of the runtime trajectory corresponding to $k=(128{,}00,128{,}100)$ are shown in (e), where the communication range of agents is shown by dashed lines.
\vspace{-0.5cm}}
}
{
  \addtolength{\tabcolsep}{-5pt}
  \begin{tabular}{p{0.32\textwidth}p{0.335\textwidth}p{0.375\textwidth}}
    \subfigure[\footnotesize Dual variables (Agent 1)]
    {
    \label{fig:numericals_dual1}
    \input{fig_dual_variable_agent1.tex}
    }
    \subfigure[\footnotesize Dual variables (Agent 2)]{
    \label{fig:numericals_dual2}
    \input{fig_dual_variable_agent2.tex}
    }
    &
    \subfigure[\footnotesize Constraint satisfaction]{
    \label{fig:numericals_constraint}
    \input{fig_variable_v.tex}
    }
    \subfigure[\footnotesize Communication]{
    \label{fig:numericals_communication}
    \input{fig_communication.tex}
    }
    &
    \vspace{1ex}
    \subfigure[\footnotesize Policies and example trajectories]{
    \label{fig:numericals_policies}
    \input{fig_trajectories.tex}
    }
  \end{tabular}
  \addtolength{\tabcolsep}{5pt}
  \vspace{-5ex}
}
\end{figure}
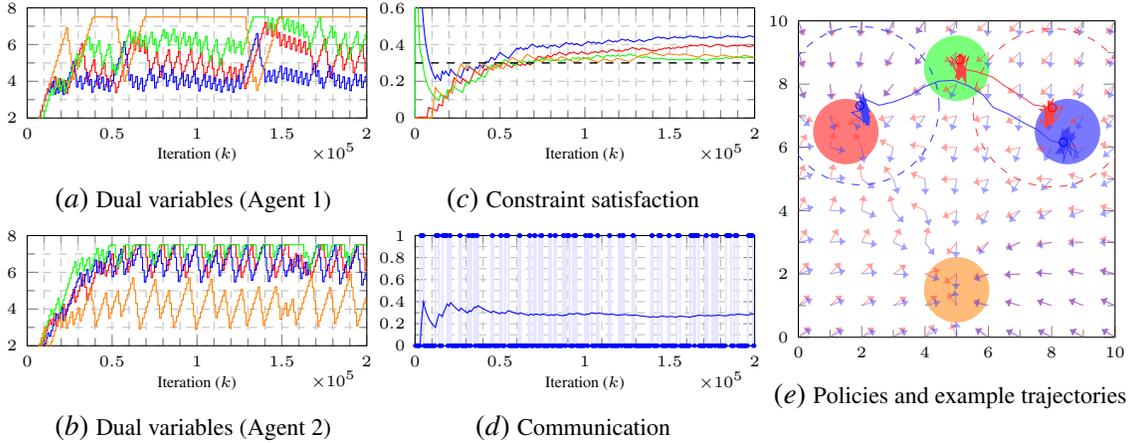
We consider a simulated experiment with $N=2$ agents monitoring the $M=4$ circular zones in Fig. \ref{fig:numericals_policies}.  The position  $S_{tn}$ of each agent  evolves as $S_{(t+1)n}=S_{tn}+T_s A_{tn}$ with $T_s=0.5$ within the area $\mathcal{S}=[0,10] \times [0,10]$ measured in squared meters. At time $t$, a reward $r_m(S_t,A_t)=1$ is obtained by the team if there is an agent inside the $m$-th region and $r_m(S_t,A_t)=0$ otherwise. By specification, each region must be visited at least $30\%$ of the time. As the requirements add up to more than $100\%$ of the time, a single agent cannot satisfy them by acting alone.
The agent policies are parameterized as linear combinations of Gaussian kernels and trained offline, running $150,000$ episodes of Actor-Critic updates using the local gradients in Proposition \ref{prop_gradients}.
Then, in the execution phase, the agents gossip their rewards within a communication range of $2.5$ meters, coordinating via Algorithm \ref{algo:alg_main_algorithm} during $200{,}000$ time steps, and updating the multipliers every $T_0=1{,}000$ steps.

\textcolor{black}{In this first experiment, we weaken the assumption of a connected graph, allowing for stochastic links.}  The agents asymptotically satisfy the requirements for all regions (Fig.~\ref{fig:numericals_constraint}), despite being within communication range only $30\%$ of the time (Fig.~\ref{fig:numericals_communication}). Coordination is achieved via the dual variables, which differ for each agent, alternating as they incorporate the rewards collected directly or via gossiping.
%
Fig.~\ref{fig:numericals_policies} shows the policies of the two agents as a function of the spatial variable for a fixed $\lambda=[5,2.5,0,5]$.  Both red and orange regions produce the same \textcolor{black}{weighted rewards} ($\lambda_1=\lambda_4=5$), yet the agents are directed to different regions, acquiring both rewards simultaneously. The trajectories of the agents are also shown between iterations $k=(128{,}000,128{,}100)$. Here, the red agent moves from the right region to the center region, while the blue agent moves from the left region to the right region. As this movement occurs, the agents enter the communication range, exchanging their stored rewards and updating their local dual variables (Fig.~\ref{fig:numericals_dual1}--\ref{fig:numericals_dual2}). This update causes the blue agent to avoid the middle region (occupied now by the red agent) and to continue to the region on the right, thus acquiring a higher reward as a team.

Next, we include two additional simulations as steps towards our future work, which aims to implement Algorithm \ref{algo:alg_main_algorithm} in a real-life environment with autonomous vehicles. In the first experiment, shown in Fig. \ref{fig:numericals_gazebo}, we tested the policy trained for the experiment in Fig. \ref{fig:numericals_policies} in a
Gazebo world, under simulated dynamics of two TurtleBot robots, including collision-avoidance. For that purpose, we coded the distributed Algorithm \ref{algo:alg_main_algorithm} as a ROS2 node communicating with the TurtleBot navigator. The second experiment in Fig. \ref{fig:numericals_floorplan} shows the trajectories of two agents
that use our multi-agent Algorithm \ref{algo:alg_main_algorithm} in an environment describing an actual floorplan, with an L-shaped corridor connecting three offices and three labs.  In this case, the agents had to address more complex dynamics and obstacles to reach the marked areas in this floorplan. But part of these dynamics is handled by a low-level SLAM and navigation controller capable of guiding each robot to a desired destination, avoiding
collisions on the way. Thus, our policies only need to assign the $M=3$ zones to the $N=2$ agents via actions in a finite space  $A_{tn}\in\{1,\ldots,M\}$. The next state observed by the policy is the position reached by the agent after navigating $T_s=1$s towards the goal. At this moment, the assignment may be changed by the policy, directing the car to an alternate zone. Fig.
\ref{fig:numericals_satisfaction} shows how constraints $c=[0.25,0.25,0.25,0.25]$ and $c=[0.5,0.4,0.3]$, which are unattainable by a single agent, can still be satisfied by our multi-agent assignment policies in these more realistic cases, reaching the specified time requirements for all the zones.
\begin{figure}[t]
\floatconts
{fig:gazebo}
{
\caption{\small (a) ROS2 Implementation of the multi-agent assignment Algorithm \ref{algo:alg_main_algorithm} in a Gazebo TurtleBot environment; (b) floorplan patrol with low-level navigation control; (c) constraint satisfaction as a function of the time step (in thousands) for the floorplan (top) and the TurtleBot (bottom).
\vspace{-0.5cm}}
}
{
  \addtolength{\tabcolsep}{-5pt}
  \begin{tabular}{p{0.27\textwidth}p{0.37\textwidth}p{0.4\textwidth}}
    \subfigure[\footnotesize Gazebo]{
    \label{fig:numericals_gazebo}
            \includegraphics[scale=0.195]{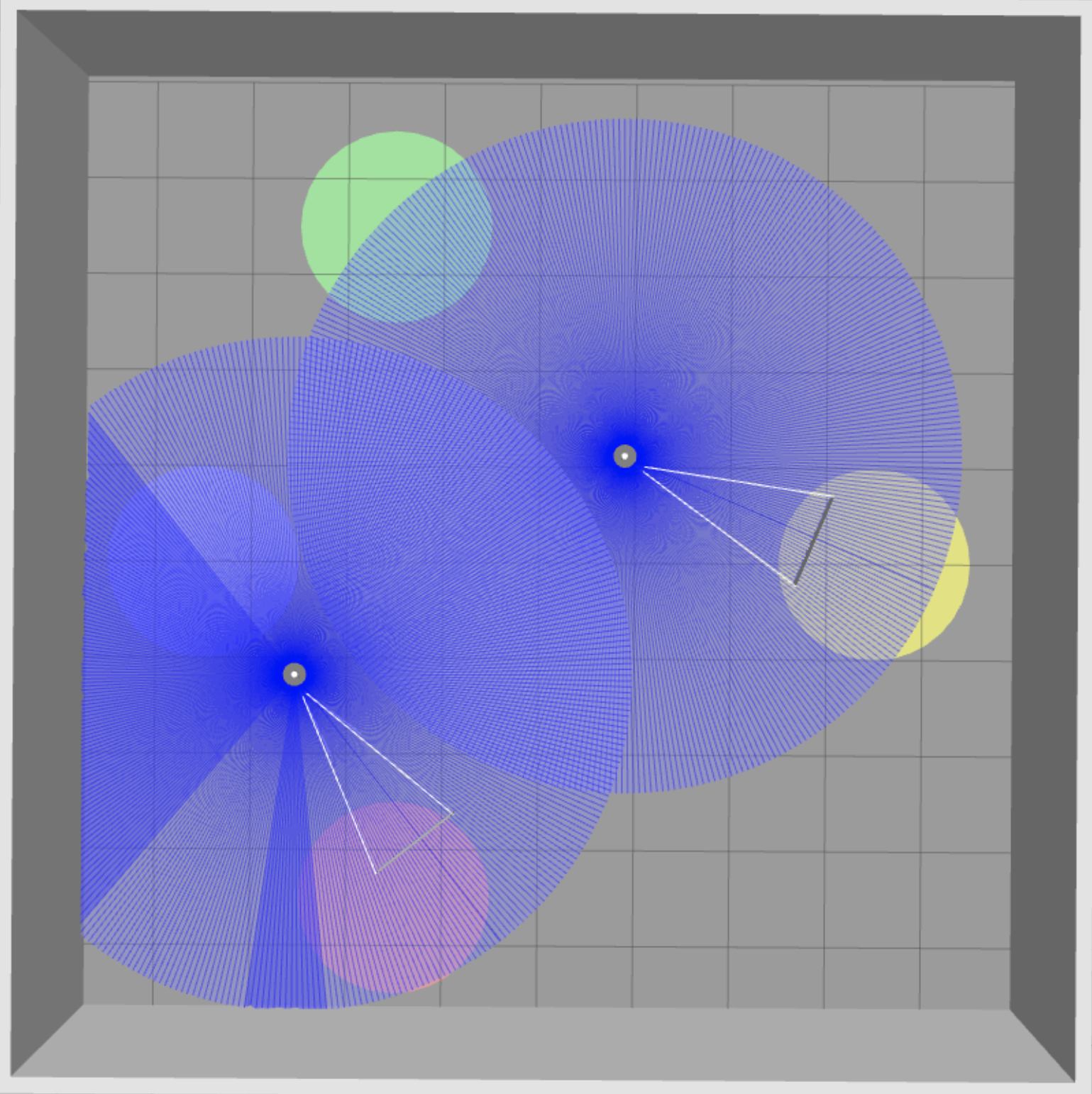}
     }
    &
     \subfigure[\footnotesize Policy for low-level navigation]{
    \label{fig:numericals_floorplan}
           \input{navigation}
     }
    &
    \subfigure[\vspace{-0.3cm}\footnotesize Constraint satisfaction]{
    \input{satisfaction}
    \label{fig:numericals_satisfaction}
    }
  \end{tabular}
  \addtolength{\tabcolsep}{5pt}
  \vspace{-5ex}
}
\end{figure}




%% file: fig_dual_variable_agent1.tex
\begin{tikzpicture}
	\pgfplotsset{grid style={dashed,gray!50}}
	\begin{axis}[
		tick scale binop=\times,
		scaled x ticks={base 10:-5},
		every x tick scale label/.style={anchor=south west,	xshift=21ex, yshift=-1.25ex},
        ticklabel style = {font=\tiny},
        label style={font=\tiny},
		minor y tick num=1,
		minor x tick num=4,
		xlabel={Iteration ($k$)},
		xlabel style={yshift=7.5pt},
		width=0.4\textwidth,
		height=0.2\textwidth,
		xmin=1,xmax=200000,ymin=2,ymax=8,
		grid=both,
	]

	\addplot [red, mark=none]
		table[x index=0,y index=1, col sep=comma]{dataDualVariableAgent1.dat};

	\addplot [blue, mark=none]
		table[x index=0,y index=2, col sep=comma]{dataDualVariableAgent1.dat};

	\addplot [green, mark=none]
		table[x index=0,y index=3, col sep=comma]{dataDualVariableAgent1.dat};

	\addplot [orange, mark=none]
		table[x index=0,y index=4, col sep=comma]{dataDualVariableAgent1.dat};

%
%
%

\end{axis}
\end{tikzpicture}

%% file: fig_dual_variable_agent2.tex
\begin{tikzpicture}
	\pgfplotsset{grid style={dashed,gray!50}}
	\begin{axis}[
		tick scale binop=\times,
		scaled x ticks={base 10:-5},
		every x tick scale label/.style={anchor=south west,	xshift=21ex, yshift=-1.25ex},
        ticklabel style = {font=\tiny},
        label style={font=\tiny},
		minor y tick num=1,
		minor x tick num=4,
		xlabel={Iteration ($k$)},
		xlabel style={yshift=7.5pt},
		width=0.4\textwidth,
		height=0.2\textwidth,
		xmin=1,xmax=200000,ymin=2,ymax=8,
		grid=both,
	]

	\addplot [red, mark=none]
		table[x index=0,y index=1, col sep=comma]{dataDualVariableAgent2.dat};

	\addplot [blue, mark=none]
		table[x index=0,y index=2, col sep=comma]{dataDualVariableAgent2.dat};

	\addplot [green, mark=none]
		table[x index=0,y index=3, col sep=comma]{dataDualVariableAgent2.dat};

	\addplot [orange, mark=none]
		table[x index=0,y index=4, col sep=comma]{dataDualVariableAgent2.dat};

%
%
%

\end{axis}
\end{tikzpicture}

%% file: fig_variable_v.tex
\begin{tikzpicture}
	\pgfplotsset{grid style={dashed,gray!50}}
	\begin{axis}[
		tick scale binop=\times,
		scaled x ticks={base 10:-5},
		every x tick scale label/.style={anchor=south west,	xshift=21ex, yshift=-4.1ex},
        ticklabel style = {font=\tiny},
        label style={font=\tiny},
		minor y tick num=1,
		minor x tick num=4,
		xlabel={Iteration ($k$)},
		xlabel style={yshift=7.5pt},
		width=0.4\textwidth,
		height=0.2\textwidth,
		xmin=1,xmax=200000,ymin=0,ymax=0.6,
		grid=both,
	]

	\addplot [red, mark=none]
		table[x index=0,y index=1, col sep=comma]{dataVarV.dat};

	\addplot [blue, mark=none]
		table[x index=0,y index=2, col sep=comma]{dataVarV.dat};

	\addplot [green, mark=none]
		table[x index=0,y index=3, col sep=comma]{dataVarV.dat};

	\addplot [orange, mark=none]
		table[x index=0,y index=4, col sep=comma]{dataVarV.dat};

	\addplot[black, dashed, line width=0.5pt] coordinates {(0,0.3) (200000,0.3)};

\end{axis}
\end{tikzpicture}

%% file: fig_communication.tex
\begin{tikzpicture}
	\pgfplotsset{grid style={dashed,gray!50}}
	\begin{axis}[
		tick scale binop=\times,
		scaled x ticks={base 10:-5},
		every x tick scale label/.style={anchor=south west,	xshift=21ex, yshift=-4.1ex},
        ticklabel style = {font=\tiny},
        label style={font=\tiny},
		minor y tick num=1,
		minor x tick num=4,
		xlabel={Iteration ($k$)},
		xlabel style={yshift=7.5pt},
		width=0.4\textwidth,
		height=0.2\textwidth,
		xmin=1,xmax=200000,ymin=0,ymax=1,
		grid=both,
	]


	\addplot [color=blue!10,solid, thick,mark=*,mark size=0.3pt,mark options={blue,solid}]
		table[x index=0,y index=1, col sep=comma]{dataCommunication.dat};

	\addplot [color=blue,solid,mark=none]
		table[x index=0,y index=2, col sep=comma]{dataCommunication.dat};

\end{axis}
\end{tikzpicture}

%% file: fig_trajectories.tex
\begin{tikzpicture}
	\begin{axis}[
		minor tick num=1,
		ylabel near ticks,
        xticklabel style = {font=\tiny},
        yticklabel style = {font=\tiny},
		width=0.38\columnwidth,
		height=0.38\columnwidth,
		xmin=0,xmax=10,ymin=0,ymax=10,
		grid=none,
	]

	\draw [opacity=0.7, draw=red!75,thick, fill=red!75] (axis cs:1.5,6.5) ellipse (10 and 10);
	\draw [opacity=0.7, draw=blue!75,thick, fill=blue!75] (axis cs:8.5,6.5) ellipse (10 and 10);
	\draw [opacity=0.7, draw=green!75,thick, fill=green!75] (axis cs:5,8.5) ellipse (10 and 10);
	\draw [opacity=0.7, draw=orange!75,thick, fill=orange!75] (axis cs:5,1.5) ellipse (10 and 10);

	\addplot[red!75, opacity=0.5,
	              point meta=sqrt(\thisrow{u}*\thisrow{u}+\thisrow{v}*\thisrow{v}),
				  quiver={
				  		u=\thisrow{u},
				  		v=\thisrow{v},
				  		scale arrows=0.5,
				  		every arrow/.append style={
				  			line width=0.3,
				  			-{Latex[scale length=0.5]}
				  		}
			    },
    			samples=21,
	] table[col sep=comma]{dataQuiverAgent1.dat};

	\addplot[blue!75, opacity=0.5,
	              point meta=sqrt(\thisrow{u}*\thisrow{u}+\thisrow{v}*\thisrow{v}),
				  quiver={
				  		u=\thisrow{u},
				  		v=\thisrow{v},
				  		scale arrows=0.5,
				  		every arrow/.append style={
				  			line width=0.3,
				  			-{Latex[scale length=0.5]}
				  		}
			    },
    			samples=21,
	] table[col sep=comma]{dataQuiverAgent2.dat};

	\addplot [color=red,mark=o,mark options={scale=0.75pt,solid,fill=none}] coordinates {(8.0153,7.2517)};
	\addplot [color=red,mark=o,mark options={scale=0.75pt,solid,fill=none}] coordinates {(5.096,8.7671)};
	\addplot [color=red!75,solid,mark=none]
		table[x index=0,y index=1, col sep=comma]{dataTrajectoryAgent1.dat};

	\addplot [color=blue,mark=o,mark options={scale=0.75pt,solid,fill=none}] coordinates {(1.9432,7.3193)};
	\addplot [color=blue,mark=o,mark options={scale=0.75pt,solid,fill=none}] coordinates {(8.3663,6.1626)};
	\addplot [color=blue!75,solid,mark=none]
		table[x index=0,y index=1, col sep=comma]{dataTrajectoryAgent2.dat};

	\draw [draw=red!75, dashed, fill=none] (axis cs:8.0153,7.2517) ellipse (25 and 25);

	\draw [draw=blue!75, dashed, fill=none] (axis cs:1.9432,7.3193) ellipse (25 and 25);

\end{axis}
\end{tikzpicture}

%% file: navigation.tex
\begin{tikzpicture}[scale=0.42]
\draw [draw=black,fill=black!20!white] (0,5.1) rectangle (10.9,6.05);
\draw [draw=black,fill=black!20!white] (10.9,0) rectangle (11.8,6.05);
\draw [draw=black!20!white,fill=black!20!white] (0,5.15) rectangle (11.5,6.0);
\draw [draw=black,fill=black!20!white] (11.8,2.6) rectangle (13.7,4.3);
\draw [draw=black,fill=black!20!white] (11.8,4.3) rectangle (13.7,6.05);
\draw [draw=black,fill=black!20!white] (11.8,0) rectangle (13.7,2.6);
\draw [draw=black!20!white,fill=black!20!white] (11.7,2.1) rectangle (11.9,2.5);

\draw [draw=white,fill=white] (0,6.1) rectangle (10,7.7);
\draw [draw=white,fill=white] (0,-1.7) rectangle (10,-0.1);

\draw [draw=black,fill=black!20!white] (0,0) rectangle (4.4,5.1);
\draw [draw=black!20!white,fill=black!20!white] (3.9,5) rectangle (4.3,5.2);

\draw [draw=black,fill=black!20!white] (4.4,0) rectangle (7.1,5.1);

\draw [draw=black,fill=black!20!white] (7.1,0) rectangle (10.9,5.1);
\draw [draw=black,fill=black!20!white] (10.9,0) rectangle (11.8,2);
\draw [draw=black!20!white,fill=black!20!white] (7.2,5) rectangle (7.6,5.2);
\draw [draw=black!20!white,fill=black!20!white] (10.8,0) rectangle (11,1.95);

\filldraw[color=black!60, fill=black!5, thick](2.7,4.2) circle (0.5);
\filldraw[color=black!60, fill=black!5, thick](9.2,4.2) circle (0.5);
\filldraw[color=black!60, fill=black!5, thick](12.9,1.3) circle (0.5);

 \pgfplotsset{grid style={dashed,gray!50} }
	\begin{axis}[
		ticks=none,
		width=\columnwidth,
		height=0.5\columnwidth,
		xmin=0,xmax=30,ymin=0,ymax=13,
		grid=both,
	]

	\addplot [color=blue,solid,thick,mark=none,mark options={blue,solid}]
		table[x index=0,y index=1, col sep=comma ]{navigation_window2000_4000.csv};
	\addplot [color=red,solid,thick,mark=none,mark options={red,solid}]
		table[x index=2,y index=3, col sep=comma ]{navigation_window2000_4000.csv};



\end{axis}

\end{tikzpicture}

%% file: satisfaction.tex
\begin{tikzpicture}[scale=0.1]
	\pgfplotsset{grid style={dashed,gray!50} }
\matrix{
 \begin{axis}[
        scaled ticks=false,
        tick label style={/pgf/number format/fixed},
	yticklabel pos=upper,
		minor tick num=1,
		xtick={0,4000,8000,12000,16000,20000},
		xticklabels={$0$,$4$,$8$,$12$,$16$,$20$},
        ticklabel style = {font=\footnotesize},
	ytick={0,0.5,1},
		yticklabels={$0$,$\frac{1}{2}$,$1$},
  width=6cm,
		height=2.75cm,
		xmin=0,xmax=20000,ymin=0,ymax=1,
		grid=both,
	]

	\addplot [color=blue,solid,thick,mark=none,mark options={blue,solid}]
		table[x index=0,y index=1, col sep=comma ]{satisfaction.csv};
	\addplot [color=red,solid,thick,mark=none,mark options={red,solid}]
		table[x index=0,y index=2, col sep=comma ]{satisfaction.csv};
	\addplot [color=black,solid,thick,mark=none,mark options={black,solid}]
		table[x index=0,y index=3, col sep=comma ]{satisfaction.csv};


\end{axis}
\\
	\begin{axis}[
        scaled ticks=false,
        tick label style={/pgf/number format/fixed},
		yticklabel pos=upper,
		minor tick num=1,
		ticklabel style = {font=\footnotesize},
        xtick={0,4000,8000,12000,16000,20000},
		xticklabels={$0$,$4$,$8$,$12$,$16$,$20$},
        ytick={0,0.5,1},
		yticklabels={$0$,$\frac{1}{2}$,$1$},
		width=6cm,
		height=2.75cm,
		xmin=0,xmax=20000,ymin=0,ymax=1,
		grid=both,
	]

	\addplot [color=blue,solid,thick,mark=none,mark options={blue,solid}]   table[x index=0,y index=1, col sep=comma ]{gazebo4zones.csv};
	\addplot [color=red,solid,thick,mark=none,mark options={red,solid}]     table[x index=0,y index=2, col sep=comma ]{gazebo4zones.csv};
	\addplot [color=black,solid,thick,mark=none,mark options={black,solid}] table[x index=0,y index=3, col sep=comma ]{gazebo4zones.csv};
	\addplot [color=black,solid,thick,mark=none,mark options={yellow,solid}]table[x index=0,y index=4, col sep=comma ]{gazebo4zones.csv};



\end{axis}
\\
};

\end{tikzpicture}

%% file: 08_appendix.tex
\begin{proof}[Proposition \ref{prop_gradients}]
 The structure $\pi_\theta(A_t\mid S_t) = \prod_{n=1}^N \pi_{\theta_n}(A_{tn}\mid S_{tn})$ yields  $   \nabla_{\theta_n}\log \pi_\theta(S_t,A_t) = \nabla_{\theta_n}\log \pi_{\theta_n}\left(A_{tn}\mid S_{tn}\right).$ Substituted in the policy gradient \cite{sutton2000policy} yields 
 $\nabla_{\theta_n} \mathcal{L(\pi_\theta,\lambda)} = \mathbb{E}_{(S_t,A_t)~\sim\rho_{\theta,\lambda}}\left[\nabla_{\theta_n }\log \pi_{\theta_n}(A_{tn}\mid S_{tn})Q_{\lambda}^{\pi_\theta}(S_t,A_t)\right]$.
 %
By the law of total expectation on $S_{tn}$,$A_{tn}$,  
%
 $ \nabla_{\theta_n} \mathcal{L(\pi_\theta,\lambda)} = \mathbb{E}_{(S_{tn},A_{tn})}\left[\mathbb{E}_{(S_t,A_t)~\sim\rho_{\theta,\lambda}}\left[\nabla_{\theta_n }\log \pi_{\theta_n}(A_{tn}\mid S_{tn})Q_{\lambda}^{\pi_\theta}(S_t,A_t)\mid S_{tn},A_{tn}\right]\right]$
 
 \noindent$ = \mathbb{E}_{S_{tn},A_{tn}}\left[\nabla_{\theta_n }\log \pi_{\theta_n}(A_{tn}\mid S_{tn})\mathbb{E}_{(S_t,A_t)~\sim\rho_{\theta,\lambda}}\left[Q_{\lambda}^{\pi_\theta}(S,A)\mid S_{tn},A_{tn}\right]\right]$  
%
and substitute \eqref{eqn_q_others}.
%
\end{proof}
\vspace{-2ex}
\begin{proof}[Lemma \ref{lemma_gossip}]
If no $S_{n\tau}\in \mathcal S_m$  at time $\tau$ then  \textcolor{black}{$\hat R_{m,\tau,n,\tau}=0$} for all $n\in V$ and thus \textcolor{black}{$\hat R_{m,\tau,n,\tau}=r_{m,\tau}=0 $} for all $t\geq \tau$ and for all $n\in V$. Otherwise, if it exists $n\in V$ such that  $S_{n,\tau}\in \mathcal S_m$, then \textcolor{black}{$\hat R_{m,\tau,n,\tau}=r_{m,\tau}=1$} for all $t\geq \tau$. 
\textcolor{black}{ Hence $\hat R_{m,\tau,i,\tau+1}=r_{m,\tau}=1$ for all $i\in \mathcal N_n$, and by induction $\hat R_{m,\tau,i,\tau+h}=r_{m,\tau}=1$ for all $h$-hop neighbors of $n$, from which the consensus follows.} 
\end{proof}
\vspace{-2ex}
\begin{proof}[Proposition \ref{prop_gossip}]
 Consider 
%
$\hat \lambda_{m,n,t}^{k+1} =\left[\hat \lambda_{m,n,t}^{k} - \frac{\eta}{T_0}  \sum_{\tau=k T_0}^{(k+1)T_0}  \left(\hat R_{m,\tau,n,t} -c_m\right) \right]_+ $
as in   Algorithm \ref{algo:alg_main_algorithm}.
Although this algorithm indicates updating the whole set of multipliers for all $k$ at each time step $t$, a simplification can be done using Lemma \ref{lemma_gossip}. Indeed, the result in Lemma \ref{lemma_gossip} ensures that the rewards $\hat R_{m,\tau,n,t}$ reach its steady state value $ r_{m,\tau}$ when $t\geq \tau +d(\mathcal G)$, so that $\hat \lambda_{m,n,t}^{k} =\hat \lambda_m^{k} 
$
for all  $t\geq t_{k+1}:= k T_0 +d(\mathcal G)$. Thus, the dual update can be simplified to 
\begin{align}\label{eqn_distributed_dual_simple}
\hat \lambda_{m,n,t}^{k+1} &=\begin{cases} \left[\hat \lambda_{m,n,t}^{k} - \frac{\eta}{T_0}   \sum_{\tau=k T_0}^{(k+1)T_0}  \left(\hat R_{m,\tau,n,t} -c_m\right) \right]_+, & \text{ for }t=kT_0,\ldots,t_{k+1}-1\\
 \hat \lambda_{m,n,t_{k+1}}^{k+1}=\hat \lambda_m^{k+1} & \text{ for }t \geq  t_{k+1} 
\end{cases}
\end{align}
With $t_{k+1}-kT_0= d(\mathcal G)$ binary $R_{m,\tau,n,t}$ in \eqref{eqn_distributed_dual_simple}, the error in $\hat \lambda_{m,n,t}^{k+1}$ is $\eta\frac{d(\mathcal G)}{T_0}$ at most.      
\end{proof}
%